%% file: main.tex
\documentclass{sig-alternate-10pt}

\usepackage{color}
\usepackage[all,color]{xy}

\DeclareMathSymbol{\N}{\mathbin}{AMSb}{"4E}
\DeclareMathSymbol{\Z}{\mathbin}{AMSb}{"5A}
\DeclareMathSymbol{\R}{\mathord}{AMSb}{"52}

\newtheorem{thm}{Theorem}[section]

\newtheorem{cor}[thm]{Corollary}
\newtheorem{prop}[thm]{Proposition}

\newcommand{\reffig}[1]{Figure~\ref{#1}}
\newcommand{\refsec}[1]{Section~\ref{#1}}

\newcommand{\refcor}[1]{Corollary~\ref{#1}}

\newcommand{\refthm}[1]{Theorem~\ref{#1}}
\newcommand{\refalg}[1]{Algorithm~\ref{#1}}
\newcommand{\refprop}[1]{Prop.~\ref{#1}}
\newcommand{\refeqn}[1]{Eqn. (\ref{#1})}

\newcommand{\ignore}[1]{}
\usepackage[titlenumbered,norelsize,ruled,vlined]{algorithm2e}
\usepackage{synttree}

\definecolor{RedLetter}{rgb}{0.63,0.165,0.163}
\definecolor{GreenLetter}{rgb}{0.165,0.63,0.163}
\newcommand{\editor}[1]{}
\newcommand{\xynode}[1]{*+[o][F-]{#1}}
\newcommand{\proj}{\_}

\begin{document}
\pagestyle{empty}
%
\conferenceinfo{ICDT}{'14 Athens, Greece}
\CopyrightYear{2014} 

\title{ Leapfrog Triejoin: A Simple, Worst-Case Optimal Join Algorithm}
%

\numberofauthors{1}
\author{
%
%
\alignauthor
Todd L. Veldhuizen\\
   \affaddr{LogicBlox Inc.}\\
   \affaddr{Two Midtown Plaza}\\
   \affaddr{1349 West Peachtree Street NW}\\
   \affaddr{Suite 1880, Atlanta GA 30309}\\
   \email{\small {\sf tveldhui@\{logicblox.com,acm.org\}}}
}
\maketitle
\begin{abstract}
\input{abstract}
\end{abstract}


\terms{Algorithms,Theory}


\input{body}

\end{document}

%% file: abstract.tex
Recent years have seen exciting developments in join algorithms.
In 2008, Atserias, Grohe and Marx (henceforth AGM) proved a tight
bound on the maximum result size of a full conjunctive query, given
constraints on the input relation sizes.   In 2012, Ngo, Porat,
R{\'e} and Rudra (henceforth NPRR) devised a join algorithm with
worst-case running time proportional to the AGM bound \cite{Ngo:PODS:2012}.
Our commercial database system LogicBlox employs a novel join
algorithm, \emph{leapfrog triejoin}, which compared conspicuously
well to the NPRR algorithm in preliminary benchmarks.  This spurred
us to analyze the complexity of leapfrog triejoin.  In this paper
we establish that leapfrog triejoin is also worst-case optimal, up
to a log factor, in the sense of NPRR.  We improve on the results
of NPRR by proving that leapfrog triejoin achieves worst-case
optimality for finer-grained classes of database instances, such
as those defined by constraints on projection cardinalities.  We
show that NPRR is \emph{not} worst-case optimal for such classes,
giving a counterexample where leapfrog triejoin runs in $O(n \log n)$
time and NPRR runs in $\Theta(n^{1.375})$ time.  On a
practical note, leapfrog triejoin can be implemented using conventional
data structures such as B-trees, and extends naturally to 
$\exists_1$ queries.   We believe our algorithm offers a useful addition to the
existing toolbox of join algorithms, being easy to absorb, simple
to implement, and having a concise optimality proof.

%% file: body.tex
\newcommand{\TJ}[1]{\editor{#1}}

\hyphenation{leap-frog}
\hyphenation{trie-join}

\section{Introduction}

Join processing is a fundamental and comprehensively-studied problem 
in database systems.  Many useful queries can be formulated as
one or more \emph{full conjunctive queries}.
A full conjunctive query is a conjunctive query
with no projections, i.e., every variable in the body appears
in the head \cite{Chandra:STOC:1977,Abiteboul:1995}.
As a running example we use the query defined by this Datalog rule:
\begin{align}
\label{e:RST}
Q(a,b,c) &\leftarrow R(a,b),S(b,c),T(a,c).
\end{align}
where $a,b,c$ are query variables (for intuition: if $R=S=T$, then $Q$ finds triangles.)

Given constraints on the sizes of the input relations
such as $|R| \leq n$, $|S| \leq n$, $|T| \leq n$,
what is the maximum possible query result size $|Q|$?
This question has practical import, since a tight
bound $|Q| \leq f(n)$ implies an
$\Omega(f(n))$ worst-case running time for algorithms
answering such queries.

Atserias, Grohe and Marx (AGM \cite{Atserias:FOCS:2008})
established a tight bound on the size of $Q$: the
\emph{fractional edge cover} bound $Q^\ast$
(\refsec{s:fractionalcover}).  For the
case where $|R|=|S|=|T|=n$, the fractional
cover bound yields $|Q| \leq Q^\ast = n^{3/2}$.
In earlier work, Grohe and Marx \cite{Grohe:SODA:2006}
gave an algorithm with running time $O(|Q^\ast|^2 g(n))$,
where $g(n)$ is a polynomial determined by the fractional
cover bound.
In 2012, 
Ngo, Porat, R{\'e} and Rudra (NPRR \cite{Ngo:PODS:2012})
devised a groundbreaking algorithm with
worst-case running time $O(Q^\ast)$, matching the AGM bound.
The algorithm is non-trivial, and its implementation and analysis depend on
rather deep machinery developed in the paper.

The NPRR algorithm was brought to our attention by Dung Nguyen,
who implemented it experimentally using our framework.
LogicBlox uses a novel
 and hitherto proprietary
join algorithm we call \emph{leapfrog triejoin}.
Preliminary benchmarks
suggested that leapfrog triejoin
performed dramatically better than NPRR on some
test problems~\cite{Nguyen:private:2012}.
These benchmark results motivated us to
analyze our algorithm, in light of the breakthroughs of NPRR.

Conventional join implementations employ
a stable of join operators (see e.g. \cite{Graefe:CSUR:1993}) which are
composed in a tree to produce the query result; this tree is
prescribed by a query plan produced by the optimizer.
The query plan often relies on producing intermediate
results.
In contrast, leapfrog triejoin joins all input
relations simultaneously without producing any
intermediate results.\footnote{In some situations it
is \emph{desirable} to materialize intermediate results.  Such materializations
are compatible with leapfrog triejoin, but are not
required to meet the worst-case performance bound, and are beyond
the scope of this paper.}  Our algorithm is variable-oriented:
for a join $Q(x_1,\ldots,x_k)$, leapfrog triejoin performs a
backtracking search, binding each variable $x_1,x_2,\ldots$
in turn to enumerate satisfying assignments of the formula defining the query.
This is in contrast to typical DBMS algorithms which are
join-oriented, using a composition of algebraic joins in a
specified order to produce the result.  Leapfrog triejoin
is substantially different than typical join algorithms,
but natural in retrospect.

In this paper we show that leapfrog triejoin
achieves running time $O(Q^\ast \log n)$, where
$Q^\ast$ is the fractional cover bound, and 
$n$ is the largest cardinality among relations of the join.
(A variant suggested by Ken Ross
eliminates the $\log n$ factor, achieving $O(Q^\ast)$ time
(\refsec{s:hashvariant}).)

We believe that leapfrog triejoin offers a useful addition to the existing
toolbox of join algorithms.  The algorithm is easy to understand
and simple to implement.
The optimality proof is concise, and
could be taught in an advanced undergraduate course. 
The optimality principle strengthens and improves that of NPRR,
and leapfrog triejoin is asymptotically faster than NPRR for
useful classes of problems.  Finally, leapfrog
triejoin is a well-tested, practical algorithm, serving as the workhorse
of our commercial database system.

The paper is organized as follows.  
In \refsec{s:background} we review the fractional edge cover
bound.  \refsec{s:entireleapfrogtriejoin} presents the
leapfrog triejoin algorithm.
\refsec{s:complexity} develops the tools used in
the complexity analysis, culminating in the
optimality proof of \refthm{thm:lftj}.
In \refsec{s:finer} we consider finer-grained
complexity classes for which leapfrog triejoin
is optimal; in one such example we demonstrate
that the NPRR algorithm has running time
$\Theta(n^{1.375})$, compared to $O(n \log n)$
for leapfrog triejoin.
In \refsec{s:extensions} we describe the
extension of leapfrog triejoin to $\exists_1$
queries.  In \refsec{s:hashvariant} we discuss
a variant of leapfrog triejoin
which eliminates the $\log n$ factor.

\section{Preliminaries and background}

\label{s:background}

\subsection{Notations and conventions}

All logarithms are base 2, and $[n] = \{ 1, \ldots, n \}$.
Complexity analyses assume the RAM machine model.

Database instances are finite structures defined
over universes which are subsets of $\N$.
In algorithm descriptions we use $\mathsf{int}$ as a synonym for $\N$.
(The restriction to $\N$ is merely to simplify the presentation;
our implementation requires only that a type be totally ordered.)

For a binary relation $R(a,b)$,
we write $R(a,\proj)$ for the projection $\pi_1(R)$, i.e., the
set $\{ a ~:~ \exists b ~.~ (a,b) \in R \}$.
For a parameter $a$, we write $R_a(b)$ for the \emph{curried}
version of $R$, i.e., the relation $\{ b ~:~ (a,b) \in R \}$.
Similarly for relations of arity $> 2$, e.g., for $S(a,b,c)$
we write $S_{a}(b,c)$ and $S_{a,b}(c)$ for curried versions.
We assume set semantics: query results are sets rather
than multisets; our datalog system implements set semantics,
unlike commercial SQL systems, so this is not merely a simplifying
convenience.

\input{bound}

\input{nprr}
\section{Leapfrog Triejoin}

\label{s:entireleapfrogtriejoin}

Leapfrog triejoin is a join algorithm for 
$\exists_1$ queries,
that is, queries definable by first-order formulae without
universal quantifiers (and, needless to say, excluding negated existential
quantifiers.)  In this paper we focus on the full conjunctive fragment
of $\exists_1$, to which our complexity bound applies.
(The additional machinery needed to go from full conjunctive queries
to $\exists_1$ is described informally in \refsec{s:extensions},
as a guide to implementors.)

In our datalog implementation, rule bodies are restricted to be
$\exists_1$ formulas.  We use leapfrog triejoin to enumerate
satisfying assignments of rule bodies.
We assume input relations are always provided in sorted order,
consistent with the data structures used by our system.
Leapfrog triejoin uses \emph{iterator interfaces} to
unify the presentation of input relations
and views of (nonmaterialized) subexpressions of a join.
A relation $A(x)$ is presented by a linear iterator,
with familiar methods such as $\mathit{next}()$ and
$\mathit{atEnd}()$, which present the elements of $A$
in order.  A disjunction such
as $A(x) \vee B(x)$ is likewise presented by a linear
iterator whose $\mathit{next}()$ method manipulates
iterators for $A,B$ to present a non-materialized view
of the disjunction.  Hence in a conjunction
$C(x),D(x)$, it does not matter whether $C$
is an input relation, or a presentation of a
non-materialized view such as $A(x)\vee B(x)$.
A similar approach is used for joins with multiple variables,
where relations and views are presented by \emph{trie iterators},
whose interface is described below.

We first describe the leapfrog join for
unary relations (\refsec{s:leapfrog}).  This is then extended to
the triejoin algorithm for full conjunctive queries
(\refsec{s:triejoin}). 
With minor embellishments,
leapfrog triejoin can tackle $\exists_1$ queries; we summarize
these in \refsec{s:extensions},
but the focus of this paper (and particularly, the
complexity analysis) is on full conjunctive queries.

\subsection{Leapfrog join for unary predicates}

\label{s:leapfrog}

The basic building block of leapfrog triejoin is a unary
join which we call \emph{leapfrog join}.
The unary leapfrog join is a variant of sort-merge join which 
simultaneously joins unary relations $A_1(x),\ldots,A_k(x)$.
The unary join is of no particular novelty 
(see e.g. \cite{Hwang:SIAMJC:1972,Demaine:SODA:2000}),
but serves as the basic building block
for leapfrog \emph{triejoin}.  Its performance bound underpins the
complexity analyses which follow.

For the purposes of leapfrog join, unary relations $A_i \subseteq \N$ are
presented in sorted order by linear iterators, one for each relation,
using this interface:

\vspace{0.5em}
\begin{centering}
\begin{tabular}{ll}
\hline
{\sf int key()}   & Returns the key at the current \\
                  & iterator position \\
{\sf next()}  & Proceeds to the next key \\
{\sf seek(int seekKey)} & Position the iterator at a least \\
                        & upper bound for seekKey, \\
                        & i.e. the least key $\geq$ seekKey, or\\
                        & move to end if no such key exists. \\
                        & The sought key must be $\geq$ the \\
                        & key at the current position. \\
{\sf bool atEnd()} & True when iterator is at the end. \\ \hline
\end{tabular}

\end{centering}
\vspace{0.5em}

\noindent
The key() and atEnd() methods are required to take $O(1)$ time,
and the next() and seek() methods are required to take $O(\log N)$
time, where $N$ is the cardinality of the relation.
Moreover, if $m$ keys are visited in ascending order,
the amortized complexity is required to be $O(1 + \log (N/m))$,
which can be accomplished using standard data structures
(notably, balanced trees such as B-trees.\footnote{For example,
if every key is visited in order
then $m=N$ and the amortized complexity is $O(1)$.  Rather than returning
to the tree root for each seek() request, the iterator ascends just
far enough to find an upper bound for the key sought.})

Leapfrog join is itself implemented as an instance of the
linear iterator interface: it provides an iterator for the
intersection $A_1 \cap \cdots \cap A_k$.  The algorithm uses
an array $\mathsf{Iter}[0 \ldots k-1]$ of pointers
to iterators, one for each relation.  In operation, the join
tracks the smallest and largest keys at which iterators
are positioned, and repeatedly moves an iterator at the smallest
key to a least upper bound for the largest key,
`leapfrogging' the iterators until they are all positioned
at the same key.  Detailed descriptions of the algorithm
follow; some readers may choose to skip to the complexity
analysis (\refsec{s:leapfrogcomplexity}).

When the leapfrog
join iterator is constructed, the leapfrog-init method
(\refalg{alg:leapfroginit}) is used
to initialize state and find the first result.
The leapfrog-init method is provided an array of
iterators; it ensures the iterators are sorted
according to the key at which they are positioned,
an invariant that is maintained throughout.

The main workhorse is leapfrog-search
(\refalg{alg:leapfrogsearch}), which finds
the next key in the intersection $A_1 \cap \cdots \cap A_k$.

\begin{algorithm}
\eIf{any iterator has $\mathsf{atEnd}()$ true}{
  $\mathit{atEnd}$ := true \;
}{
  $\mathit{atEnd}$ := false \;
  sort the array $\mathsf{Iter}[0..\mathrm{k}-1]$ by keys at which the iterators are positioned \;
  $p$ := 0 \;
  leapfrog-search()
}
\caption{\label{alg:leapfroginit}leapfrog-init()}
\end{algorithm}

\begin{algorithm} 
$x'$ := $\mathsf{Iter}[(p-1) \mod k].\mathsf{key}()$ \tcp*{Max key of any iter} 
\While{true}{
  $x$ := $\mathsf{Iter}[p].\mathsf{key}()$           \tcp*{Least key of any iter}
  \eIf{$x=x'$}{
   $\mathit{key}$ := $x$                                      \tcp*{All iters at same key}
   return\;
  }{
   $\mathsf{Iter}[p].\mathsf{seek}(x')$\;
   \eIf{$\mathsf{Iter}[p].\mathsf{atEnd}()$}{
     $\mathit{atEnd}$ := true \;
     return\;
   }{
     $x'$ := $\mathsf{Iter}[p].\mathsf{key}()$\;
     $p$ := $p+1 \mod k$\;
   }
  }
}
\caption{\label{alg:leapfrogsearch}leapfrog-search()}
\end{algorithm}

Immediately following leapfrog-init(), the leapfrog join iterator is
positioned at the first result, if any; subsequent results are obtained
by calling leapfrog-next() (\refalg{alg:leapfrognext}).
To complete the linear iterator interface,
we define a leapfrog-seek() function which finds the first element of $R_1 \cap \cdots \cap R_k$
which is $\geq $ seekKey (\refalg{alg:leapfrogseek}).

\begin{algorithm} 
\caption{\label{alg:leapfrognext}leapfrog-next()}
$\mathsf{Iter}[p].\mathsf{next}()$\;
\eIf{$\mathsf{Iter}[p].\mathsf{atEnd}()$}{
   $\mathit{atEnd}$ := true\;
   }{
   $p$ := $p+1 \mod k$\;
   leapfrog-search()\;
   }
\end{algorithm}

\reffig{f:leapfrogexample} illustrates a join of three relations.

\begin{algorithm}
\caption{\label{alg:leapfrogseek}leapfrog-seek($\mathsf{int}$ seekKey)}
$\mathsf{Iter}[p].\mathsf{seek}(seekKey)$\;
\eIf{$\mathsf{Iter}[p].\mathsf{atEnd}()$}{
   $\mathit{atEnd}$ := true\;
   }{
   $p$ := $p+1 \mod k$\;
   leapfrog-search()\;
   }
\end{algorithm}

\begin{figure*}
\small
\begin{align*}
\xymatrix @=0.6cm {
A & 0 \ar@/^1pc/[rrr]^{seek(2)} & 1 & & 3 \ar@/^1pc/[rrrrr]^{seek(8)} & 4 & 5 & 6 & 7 & 8 \ar@{=}[d] \ar@/^1pc/[rrr]^{seek(10)} & 9 & & 11 \\
B & 0 \ar@/^1pc/[rrrrrr]^{seek(3)} & & 2 & & & & 6 \ar@/^1pc/[rr]^{seek(8)} & 7 & 8 \ar@{=}[d] \ar@/^1pc/[rrrr]^{seek(11)} & 9 & & & +\infty \\
C & & & 2 \ar@/^1pc/[rrrrrr]^{seek(6)} & & 4 & 5 & & & 8 \ar@{=}[d] \ar[rr]^{next()} & & 10 \\
A \cap B \cap C & & & & & & & & & 8
}
\end{align*}
\caption{\label{f:leapfrogexample}\small
Example of a leapfrog join of three relations $A,B,C$, with
$A=\{0,1,3,4,5,6,7,8,9,11\}$ and $B$, $C$ as shown in the second and
third rows.  Initially the
iterators for $A, B, C$ are positioned (respectively) at 0, 0, and 2.
The iterator for $A$ performs a seek(2) which lands it at 3;
the iterator for $B$ then performs a seek(3) which lands at 6;
the iterator for $C$ does seek(6) which lands at 8, etc.
}
\end{figure*}
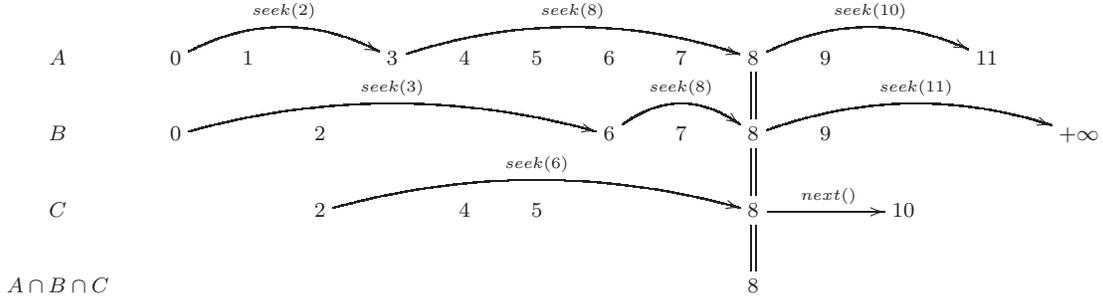

\subsection{Complexity of leapfrog join}

\label{s:leapfrogcomplexity}

In the analyses which follow, we focus on data
complexity~\cite{Vardi:STOC:1982}, i.e., we assume the query definition
to be fixed, and omit constant factors
which depend only on the structure of the query
(e.g. number of atoms and variables).

Let $N_{min} = \min \{ |A_1|, \ldots, |A_k| \}$ be the cardinality
of the smallest relation in the join,
and $N_{max} = \max \{ |A_1|,$ $\ldots, |A_k| \}$ the largest.
\begin{prop}
\label{prop:leapfrogtime}
The running time of leapfrog join is $O\left(N_{min} \log \left( N_{max}/N_{min} \right) \right)$.
\end{prop}
\begin{proof}
The leapfrog algorithm advances the iterators in a fixed pattern:
each iterator is advanced every $k$ steps of the algorithm.
An iterator for a relation with cardinality $N$ can be advanced
at most $N$ times before reaching the end;
therefore the number of steps is at most $k \cdot N_{min}$.
An iterator which visits $m$ of $N$ values in order
is stipulated to have amortized cost $O(1 + \log(N/m))$; the iterator
for a largest relation will have $N=N_{max}$ and $m=N_{min}$,
for total cost $N_{min} \cdot O(1 + \log(N_{max} / N_{min}))$.
\end{proof}

\label{s:densityexample}
The leapfrog join is able to do substantially better than pairwise
joins in some scenarios.  Suppose we have relations
$A,B,C$ where $A=\{0,\cdots,2n-1\}$,
$B=\{n,\cdots,3n-1\}$, and $C=\{0,\cdots,n-1,2n,\ldots,3n-1\}$.
Any pairwise join will produce $n$ results, 
but the intersection $A \cap B \cap C$ is empty; the leapfrog
join determines this in $O(1)$ steps.

\subsection{Trie iterators}
\label{s:trieiterator}
We extend the linear iterator interface to handle relations of arity
$> 1$.
Relations such as $A(x,y,z)$ are presented
as \emph{tries} with each tuple $(x,y,z) \in A$ corresponding to a unique path
through the trie from the root to a leaf (\reffig{f:treeiterator}).
(Note however that relations need not be stored as tries; 
in practice we use B-tree-like data structures, and 
present their contents via a trie iterator interface.)

Upon initialization, trie iterators are positioned at the root.
The linear iterator API is augmented with two methods for trie-navigation:

\vspace{0.5em}
\begin{centering}
\begin{tabular}{ll}
\hline
void open();    & Proceed to the first key at the\\
                & next depth \\
void up();      & Return to the parent key at the \\
                & previous depth \\
\hline
\end{tabular}

\end{centering}

A trie iterator for a materialized relation is required to
have $O(\log N)$ time for the open() and up() methods.

With a bit of bookkeeping, it is straightforward to
present Btree-like data structures as TrieIterators,
with each operation taking $O(\log N)$ time.\footnote{
For example, to perform a next() operation when positioned at the
node $x=1$ of \reffig{f:treeiterator}, one would seek the least upper bound of
$(1,+\infty,+\infty)$ in the B-tree representation; this would reach
the record $(3,5,2)$.}


\begin{figure}
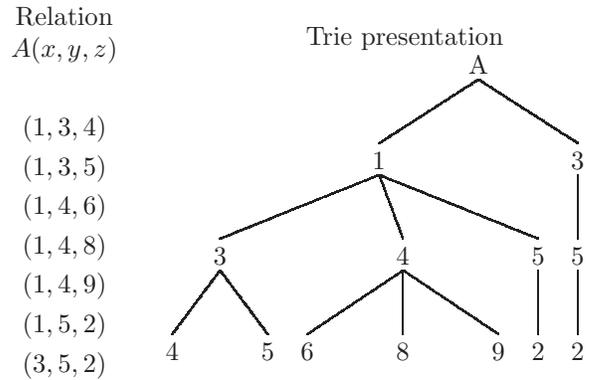

\begin{minipage}{1in}
\begin{centering}
Relation $A(x,y,z)$

\end{centering}

\begin{align*}
(1,3,4) \\
(1,3,5) \\
(1,4,6) \\
(1,4,8) \\
(1,4,9) \\
(1,5,2) \\
(3,5,2)
\end{align*}
\end{minipage}
\begin{minipage}{2.5in}
\begin{centering}
Trie presentation

\end{centering}

\synttree[A[1[3[4][5]][4[6][8][9]][5[2]]][3[5[2]]]]
\end{minipage}
\caption{\label{f:treeiterator}
Example: Trie presentation of a relation $A(x,y,z)$.
After open() is invoked at some node $n$, the linear iterator methods
next(), seek() and atEnd()
present the children of $n$.  In the above example,
invoking open() thrice on an iterator positioned at $A$
would move to the leaf node $[1,3,4]$; next() would then
move to leaf node $[1,3,5]$; another next() would result
in the iterator being atEnd().  The sequence up(), next(),
open() would then advance the iterator to the leaf node $[1,4,6]$.
}
\end{figure}

\subsection{Leapfrog Triejoin}

\label{s:triejoin}

We now describe the \emph{Leapfrog Triejoin} algorithm for full
conjunctive joins.

The triejoin algorithm requires the optimizer to choose a
\emph{variable ordering}, i.e., some permutation of the
variables appearing in the join.  For example, in the
join $R(a,b),S(b,c),T(a,c)$ we might choose the variable
ordering $[a,b,c]$.
Choosing a good variable ordering is crucial for
performance, in practice, but immaterial for the worst-case
complexity analysis presented here.  Techniques for choosing
an advantageous variable ordering are the subject of a forthcoming paper;
for the complexity analysis we fix an arbitrary ordering.

Leapfrog triejoin requires a restricted form of conjunctive joins,
attained via some simple rewrites:
\begin{enumerate}
\item Each variable can appear at most once in each
argument list.  For example, $R(x,x)$ would be rewritten
to $R(x,y), x=y$ to satisfy this requirement.  The
$x=y$ term may be presented as a nonmaterialized view
of a predicate $\mathit{Id}(x,y) \Leftrightarrow (x=y)$,
implemented by a variant of the TrieIterator interface.
\item Each argument list must be a subsequence of the
variable ordering.  For example, if the chosen variable
ordering were $[a,b,c]$ and the join contained a term
$U(c,a)$, we would rewrite this to $U'(a,c)$ and
define a materialized view $U'(a,c)\equiv U(c,a)$. (In practice
we install indices automatically when required by such rewrites,
and maintain them for use in future queries.)
%
\item To simplify the complexity analysis, each relation symbol 
may appear at most once in the query.  For a query such as 
$E(x,y)$, $E(y,z)$ we introduce a copy $E' \equiv E$ and rewrite
to $E(x,y)$, $E'(y,z)$.  This avoids awkwardness in the complexity
analysis, but is not required for implementation purposes.
\item Constants may not appear in argument lists.
A subformula such as $A(x,2)$ is rewritten to $A(x,y)$, $\mathsf{Const}_2(y)$,
where $\mathsf{Const}_2 = \{ 2 \}$.  In practice $\mathsf{Const}_\alpha$
is presented as a nonmaterialized view, using a variant of the TrieIterator
interface.
\end{enumerate}

Leapfrog triejoin employs one leapfrog join for each variable.
Consider the example $R(a,b)$, $S(b,c)$, $T(a,c)$
with the variable ordering $[a,b,c]$.  The leapfrog joins employed for the
variables $a,b,c$ in the example are:\footnote{Recall that
$R(a,\proj)$ is the projection $\{ a ~:~ \exists b ~.~ (a,b) \in R\}$,
and $R_a(b)$ is the 'curried' form $\{ b ~:~ (a,b) \in R \}$.
}

\vspace{0.5em}
\begin{tabular}{cll}
Variable & Leapfrog join & Remarks \\ \hline
$a$ & $R(a,\proj),T(a,\proj)$ & Finds $a$ present in $R$, $T$ \\
    &                         & projections \\
$b$ & $R_a(b),S(b,\proj)$ & For specific $a$, finds \\
    &                     & $b$ values \\
$c$ & $S_b(c),T_a(c)$ & For specific $a,b$, finds \\
    &                 & $c$ values
\end{tabular}
\vspace{0.5em}

\noindent
The topmost leapfrog join iterates values for $a$ which are in both the
projections $R(a,\proj)$ and $T(a,\proj)$.  When this leapfrog join
emits a binding for $a$, we can proceed to the next level join
and seek bindings for $b$ from $R_a(b),S(b,\proj)$.  For each such $b$,
we can proceed to the next level and seek a binding for
$c$ in $S_b(c),T_a(c)$.  When a leapfrog join exhausts its bindings,
we can retreat to the previous level and seek another binding for
the previous variable.
Conceptually, we can regard triejoin as a backtracking search through
a `binding trie.'


\ignore{***
\editor{This paragraph is redundant.}
It profits us to define triejoin as an implementation
of the trie iterator interface.  That is, triejoin presents
a nonmaterialized view of the join result, so
result tuples can be retrieved by exploring the trie
using the open(), next(), etc. methods.
We can then define a variant of triejoin for disjunctive joins,
which also implements the trie iterator interface; this
lets us build trie iterators for queries containing arbitrary
nestings of disjunction and conjunction.  Adding trie iterators
for complements and projections completes the toolbox needed to
tackle $\exists_1$ queries.
***}

\subsection{Triejoin implementation}

\label{s:triejoinimplementation}

At initialization, the triejoin is provided with a trie iterator
for each relation (or more generally, subformula) of the join.

The triejoin initialization
constructs an array of leapfrog join instances,
one for each variable.  The leapfrog join for a variable $x$ is
given an array of pointers to trie-iterators, one for each atom in
which $x$ appears.
For example, in the join $R(a,b)$, $S(b,c)$, $T(a,c)$, the
leapfrog join for $b$ is given pointers to the trie-iterators
for $R$ and $S$.  There is only one instance of the trie-iterator
for $R$, which is shared by the leapfrog joins for $a$ and $b$.

The leapfrog joins use the linear-iterator portion of the
trie iterator interfaces; the open/up trie navigation methods are
used only by the triejoin algorithm.
The triejoin uses a variable \emph{depth} to track the
current variable for which a binding is being sought; initially
$depth=-1$ to indicate the triejoin is positioned at the root
of the binding trie (i.e., before the first variable.)
Depths $0,1,\ldots$ refer to the first, second, etc. variables
of the variable ordering.

\ignore{***
\begin{figure}
\begin{centering}
\includegraphics[width=3in]{RSTIterExample.pdf}

\end{centering}
\caption{Example for $R(a,b)$, $S(b,c)$, $T(a,c)$ with variable
ordering [a,b,c] of the relationship between the iterator
arrays of the leapfrog
joins for each variable, and the trie iterators for R, S, and
T.}
\end{figure}
\TJ{In figure: more italics here.  Also: TJ found this hard to interpret
at a glance.  Maybe say why this figure is here: so you understand
there is just one instance of each iterator, shared among the
leapfrog joins.}
***}

Leapfrog triejoin presents a nonmaterialized view of the
query result, presented via a trie-iterator interface.
The linear iterator portions of the trie-iterator interface
(namely key(), atEnd(), next(), and seek())
are delegated to the leapfrog join for the current
variable.  (At depth -1, i.e., the root, only the operation
open() is permitted, which moves to the first variable.)  It remains
to define the open() and up() methods, which are trivial
(Algorithms \ref{alg:triejoinopen} and \ref{alg:triejoinup}).

\begin{algorithm}
\caption{\label{alg:triejoinopen}triejoin-open()}
\SetAlgoLined
\tcp{Advance to next var}
$depth$ := $depth$ + 1 \;
\For{each iter in leapfrog join at current depth}{
   iter.open() \;
}
call leapfrog-init() for leapfrog join at current depth
\end{algorithm}

\begin{algorithm}
\caption{\label{alg:triejoinup}triejoin-up()}
\SetAlgoLined
\For{each iter in leapfrog join at current depth}{
   iter.up() \;
}
\tcp{Backtrack to previous var}
$depth$ := $depth$ - 1 \;
\end{algorithm}

This completes the trie iterator interface.  To obtain the
satisfying assignments of the query formula, we simply walk the trie
presented by leapfrog triejoin, a simple exercise we omit here.



\section{Complexity Analysis}

\label{s:complexity}

We consider now the complexity of leapfrog triejoin for full conjunctive
joins of materialized relations.

\subsection{The proof strategy}

We introduce the proof strategy informally, 
before proceeding to the formal proof
of \refthm{thm:lftj}.
Consider the example join:
\begin{align*}
Q(a,b,c) \equiv R(a,b), S(b,c), T(a,c)
\end{align*}
with variable ordering $[a,b,c]$.
Suppose that $|R| \leq n$,
$|S| \leq n$, and $|T| \leq n$.
The fractional cover bound yields $|Q| \leq n^{3/2}$.

We wish to show that the triejoin runs in $O(n^{3/2} \log n)$
time for this example.
Recall that a leapfrog join of two unary relations $U,V$
requires at most $2 \cdot \min\{ |U|, |V| \}$ iterator operations.
It is readily seen that the cost at the first two trie levels $[a,b]$ 
cannot exceed $O(n)$ linear iterator operations: at the first trie level
the leapfrog join is limited by $\min(|R(a,\proj)|,|T(a,\proj)|) \leq |R(a,\proj)| \leq |R| \leq n$,
and at the second trie level the number of iterator operations
is controlled by:
\begin{align*}
& \sum_{a \in R(a,\proj),T(a,\proj)} \min \{ |R_a(b)|,|S(b,\proj)| \} \\
& \leq \sum_{a \in R(a,\proj),T(a,\proj)} |R_a(b)| \\
& \leq |R|
\end{align*}
Therefore the total number of linear iterator operations at the first two
trie levels is $O(n)$.
At the third trie level, the number of linear iterator operations
is controlled by:
\begin{align}
\sum_{(a,b) \in R(a,b),S(b,\proj),T(a,\proj)}
\min \{|S_b(c)|, |T_a(c)| \}
\label{e:thirdtrielevel}
\end{align}

We now wish to show that the quantity (\ref{e:thirdtrielevel}) is $\leq n^{3/2}$.
We do this by renumbering the $c$ values of the relations such that the join
produces a number of results equal to (\ref{e:thirdtrielevel}), without
increasing the amount of work.
Since the join can produce at most $n^{3/2}$ results (from the
fractional cover bound), this will establish that (\ref{e:thirdtrielevel})
is $\leq n^{3/2}$.

For a concrete example, suppose we had these trie presentations of $R,S,T$:
{\small
\begin{align*}
\xymatrix @=0.25cm {
& R(a,b) \ar@{-}[d] & & S(b,c) \ar@{-}[dd] & & & & & T(a,c) \ar@{-}[d] \\
a & 7 \ar@{-}[d] & & & & & & & 7 \ar@{-}[dd] \ar@{-}[ddr] \ar@{-}[ddrr] \\
b & 4 & & 4 \ar@{-}[d] \ar@{-}[dr] \ar@{-}[drr] \ar@{-}[drrr] \\
c & & & 1 & 4 & 5 & 9 & & 2 & 3 & 5 
}
\end{align*}
}

This would produce only the result tuple $(7,4,5)$.  To obtain a result
size equal to (\ref{e:thirdtrielevel}) we renumber the $c$
values, resulting in a new problem instance which produces
one result for every leaf of $T$ (the
smaller relation):
{\small
\begin{align*}
\xymatrix @=0.25cm {
& R(a,b) \ar@{-}[d] & & S(b,c) \ar@{-}[dd] & & & & & T(a,c) \ar@{-}[d] \\a & 7 \ar@{-}[d] & & & & & & & 7 \ar@{-}[dd] \ar@{-}[ddr] \ar@{-}[ddrr] \\
b & 4 & & 4 \ar@{-}[d] \ar@{-}[dr] \ar@{-}[drr] \ar@{-}[drrr] \\
c & & & 0 & 1 & 2 & 3 & & 0 & 1 & 2
}
\end{align*}
}

This results in exactly three results $(7,4,0)$, $(7,4,1)$,
and $(7,4,2)$, equalling (\ref{e:thirdtrielevel}).

In general, the renumbering produces modified relations $S',T'$ which each have
cardinality $\leq n$.
Since $n^{3/2}$ is an upper bound on the result size, it follows
that (\ref{e:thirdtrielevel}) is at most $n^{3/2}$.

The renumbering is accomplished as follows:
\begin{enumerate}
\item[(i)] Construct $S'(b,c)$ by renumbering the $c$ values
of each $S_b$-subtree to be $0,1,\ldots$, i.e.:
\begin{align*}
S'(b,\proj) &= S(b,\proj) ~~~\text{\footnotesize (Keep \emph{b} values the same)} \\
S'_b &= \{ 0,1,\ldots,|S_b|-1 \} ~~~\text{\footnotesize (Renumber \emph{c} values)}
\end{align*}
\item[(ii)] Similarly, renumber the $c$ values of each $T_a$ subtree:
\begin{align*}
T'(a,\proj) &= T(a,\proj) ~~~\text{\footnotesize (Keep \emph{a} values the same)} \\
T'_a &= \{ 0,1,\ldots,|T_a|-1 \} ~~~\text{\footnotesize (Renumber \emph{c} values)}
\end{align*}
\end{enumerate}
When we compute the leapfrog join of $S'_b = \{ 0,1,\ldots,$ $|S_b|-1 \}$
with $T'_a = \{ 0,1,\ldots,|T_a|-1 \}$,
we get exactly $\min \{ |S_b|,$ $|T_a| \}$
results.  This holds for every join at the third trie level;
therefore
the query result size is exactly the quantity (\ref{e:thirdtrielevel}).
Since the fractional cover bound gives an upper bound of $n^{3/2}$
on the query result size, we have:
\begin{align*}
\sum_{(a,b) \in R(a,b),S(b,\proj),T(a,\proj)}
\min \{|S_b(c)|, |T_a(c)| \}
&\leq n^{3/2}
\end{align*}
Hence the running time of leapfrog triejoin for the example is
$O(n^{3/2} \log n)$.


The above example illustrates the proof technique we employ
for the leapfrog triejoin complexity analysis.
The following sections generalize the
renumbering transform (\refsec{s:renumbering}), 
develop the sum-min cost bound (\refsec{s:costs}),
and formalize classes of databases to which the complexity
bound applies (\refsec{s:families}).  These lead up to the proof,
in \refsec{s:lftjproof}, of the complexity bound for leapfrog triejoin
(\refthm{thm:lftj}).

\subsection{The renumbering transform}

\label{s:renumbering}

We generalize the renumbering transformation introduced in
the previous section.  For an atom $R(x,y,z)$,
\emph{a renumbering at variable $v \in \{x,y,z\}$} is obtained by
traversing the trie representation of $R$,
and:
\begin{itemize}
\item If the variable $v$ appears in the argument list
at depth $d$, then for each node at depth $d-1$ renumber its children to
be $0,1,\ldots$; otherwise, do nothing.
\item Replace all values for variables appearing after $v$ in
the key-ordering with $0$.
\item Eliminate any duplicate tuples.
\end{itemize}
The resulting relation $R'$ is called a
renumbering of $R$.
\reffig{f:renumbering} illustrates renumberings of
a relation $R(x,y,z)$ at various depths.\footnote{
It is worth noting that alternate renumbering
schemes might be useful for certain classes of
problems.  But for the purposes of this paper,
we stick to $0,1,2,\ldots$.}


\begin{figure*}
\begin{tabular}{cc}
\begin{minipage}{3in}
\begin{centering}
\synttree[r[1[3[4][5]][4[6][8][9]][5[2]]][3[5[2]]]]

\vspace{0.1in}
(a) A relation $R(x,y,z)$

\end{centering}

\end{minipage} &
\begin{minipage}{3in}
\begin{centering}
\synttree[r[1[3[0][1]][4[0][1][2]][5[0]]][3[5[0]]]]

\vspace{0.1in}
(b) Renumbered at depth $2$, for variable $z$

\end{centering}

\end{minipage} \\
\\
\begin{minipage}{3in}
\begin{centering}
\synttree[r[1[0[0]][1[0]][2[0]]][3[0[0]]]]

\vspace{0.1in}
(c) Renumbered at depth $1$, for variable $y$

\end{centering}
\end{minipage} &
\begin{minipage}{3in}
\begin{centering}
\synttree[r[0[0[0]]][1[0[0]]]]

\vspace{0.1in}
(d) Renumbered at depth $0$, for variable $x$
\end{centering}
\end{minipage}
\end{tabular}
\caption{\label{f:renumbering}Example of the renumbering transform
applied to a relation $R(x,y,z)$.}
\end{figure*}

\subsection{Triejoin costs}

\label{s:costs}
Let $R^1,\ldots,R^m$ be the relations in the join, and
$V=[v_0,\ldots,v_{k-1}]$ be the chosen variable ordering.
Each atom (relation) in the join takes as arguments some subset of the variables $V$,
in order.  For a relation $R(v_0,v_1,$ $v_2,v_3)$,
we use this notation for currying:
\begin{align*}
R_{v_0,v_1}(v_2,v_3) &= \{ (v_2,v_3) ~:~ (v_0,v_1,v_2,v_3) \in R \}
\end{align*}
We write $R_{< i}(v_i,\ldots)$
for the curried version of all variables strictly before
$v_i$ in the ordering; e.g. $R_{<3}(v_3) = R_{v_0,v_1,v_2}(v_3)$.

Write $Q_i(v_0,v_1,\ldots,v_i)$ for the join `up to and including' variable $v_i$;
this is obtained by replacing variables $v_{i+1},\ldots,v_{k-1}$ with
the projection symbol $\proj$ in the query, and omitting any
atoms which contain only projection symbols.\footnote{Note that $Q_i$ is generally a strict superset of
the projection of the query result $Q(v_0,v_1,\ldots,v_i,\proj,\proj,\ldots,\proj)$.}
For example, with $Q=R(a,b),S(b,c),T(a,c)$, and key order $[a,b,c]$,
we would have:
\begin{align*}
Q_0 &= R(a,\proj),T(a,\proj) \\
Q_1 &= R(a,b),S(b,\proj),T(a,\proj) \\
Q_2 &= R(a,b),S(b,c),T(a,c)
\end{align*}

Let $R^\alpha_{< i}(v_i,\ldots), R^\beta_{< i}(v_i,\ldots), \cdots$ be the
relations in the leapfrog join at depth $i$.
Let $C_i$ be the sum-min of the leapfrog triejoin at tree depth $i$:
\begin{align*}
C_i &= \sum_{(v_0,\ldots,v_{i-1}) \in Q_{i-1}} \min 
\begin{array}[t]{l}
\{ |R^\alpha_{< i}(v_i,\proj,\ldots,\proj)|, \\
~~~~|R^\beta_{< i}(v_i,\proj,\ldots,\proj)|, \cdots \}
\end{array}
\label{e:triejoinsummin}
\end{align*}

To compute the join result, one uses the trie iterator presented by
leapfrog triejoin to completely traverse the trie.

\begin{prop}
\label{prop:triejoincosts}
The running time of leapfrog triejoin is
$O((\sum_{i=0}^{k-1} C_i) \log N_{max})$.
\end{prop}

\begin{proof}
Let $N_{max}$ be the
cardinality of the largest input relation.  At levels $0,\ldots,k-2$,
each result of a leapfrog join incurs the cost of an open() and up()
operation, totalling $O((\sum_{i=0}^{k-2} C_i) \log N_{max})$ time by
the $O(\log N)$ performance requirement for open() and up().
At levels $0,\ldots,k-1$ the time cost of the leapfrog joins is
$O((\sum_{i=0}^{k-1} C_i) \log N_{max})$ from \refprop{prop:leapfrogtime}.
\end{proof}

\subsection{Families of problem instances}

\label{s:families}

We now formalize some concepts in preparation for asymptotic
arguments.  Chief among these is a \emph{family of problem
instances}; this concept encompasses familiar examples such as
\emph{graphs with at most $n$ edges},
and \emph{binary relations R,S,T with each relation
of size $\leq n$.}


We write $\mathrm{Str}[\sigma]$ for finite structures with
signature (vocabulary) $\sigma$.  A \emph{family of problem instances}
is a countable set
$(\mathbf{K}_n)_{n \in \N}$ indexed by a parameter $n \in \N$,
where each $\mathbf{K}_n \subseteq \mathrm{Str}[\sigma]$ is a class of
finite relational structures, and
$(i \leq j) \Longrightarrow (\mathbf{K}_i \subseteq \mathbf{K}_j)$.
(Example: \emph{graphs with at most $n$ edges} is
a family of problem instances.)

More generally, we can choose a tuple of parameters 
$\overline{n} = [n_1,\ldots,n_k] \in \N^k$,
with the usual partial ordering on tuples, such
that if $\overline{n}' = [n_1',\ldots,n_k']$ and
$n_1 \leq n_1', \ldots, n_k \leq n_k'$, then
$\mathbf{K}_{[n_1,\ldots,n_k]} \subseteq \mathbf{K}_{[n_1',\ldots,n_k']}$.
(Example: let $\sigma$ contain the binary relation symbols 
$R,S,T$, and define
$\mathbf{K}_{r,s,t}$ to be structures
with $|R| \leq r$, $|S| \leq s$, and $|T| \leq t$.)

A query $Q$ is defined by some first-order formula
$\varphi(\overline{x})$.  
For a structure $\mathcal{A} \in \mathbf{K}_n$
we write $Q^\mathcal{A}$ to mean the satisfying assignments
of $\varphi(\overline{x})$ in $\mathcal{A}$.
For simplicity, we take $\overline{x}$ to be the variable
ordering for the triejoin.

\subsection{Proof of the complexity bound}

\label{s:lftjproof}

Fix a variable ordering V.
Given structures $\mathcal{A},\mathcal{A}'$, we say $\mathcal{A}'$ is
a renumbering of $\mathcal{A}$ if it is obtained
by selecting some relation of $\mathcal{A}$ and renumbering it at some depth,
as per \refsec{s:renumbering}.
A family of problem instances is \emph{closed under renumbering}
when for every $\mathcal{A} \in \mathbf{K}_n$, if $\mathcal{A}'$
is a renumbering of $\mathcal{A}$, 
then $\mathcal{A}' \in \mathbf{K}_n$ also.

\begin{thm}
\label{thm:lftj}
Let $Q(v_0,\ldots,v_{k-1})$ be a full conjunctive query satisfying
the syntactic restrictions of \refsec{s:triejoin}, and
\begin{enumerate}
\item $(\mathbf{K}_n)_{n \in \N}$ be a family of problem instances
closed under renumbering,
\item $q(n)=\max_{\mathcal{A} \in \mathbf{K}_n} |Q^\mathcal{A}|$
be the largest query result size for any structure in $\mathbf{K}_n$, and
\item $M(n)$ be the cardinality of the largest relation
in any structure of $\mathbf{K}_n$.
\end{enumerate}
Then, Leapfrog Triejoin computes $Q$ in $O(q(n) \log M(n))$ time
over $(\mathbf{K}_n)_{n \in \N}$ using variable ordering
$[v_0,\ldots,v_{k-1}]$.
\end{thm}

\begin{proof}
(By contradiction).
Suppose the running time is not $O(q(n) \log M(n))$.
From \refprop{prop:triejoincosts}, the running time of leapfrog triejoin 
is $O((C_0 + \cdots + C_{k-1})$ $\log M(n))$, where
$C_i$ is the sum-min for the leapfrog join of variable $v_i$.
For this to not be $O(q(n) \log M(n))$,
some variable $v_i$ must have $C_i \in \omega(q(n))$ for infinitely many
instances $\mathcal{A}$.
For each such $\mathcal{A}$, renumber all relations for variable $v_i$,
and revise $Q(v_0,\ldots,v_{k-1})$ appropriately.  This results in structures
$\mathcal{A}'$ with $|Q^{\mathcal{A}'}| = C_i$.
Since the family is closed under renumbering, 
$\mathcal{A}' \in \mathbf{K}_n$; but $|Q^{\mathcal{A}'}| \in \omega(q(n))$,
contradicting the definition of $q(n)$.
\end{proof}

Note that \refthm{thm:lftj} does not depend on the
fractional edge cover bound (\refsec{s:fractionalcover}).
The fractional edge cover bound provides a means to
bound $q(n)$ for families of problem instances
defined by constraints on the size of input relations.
Since renumbering does not increase the sizes of
relations, such families are closed under renumbering.
The worst-case optimality in the sense of NPRR
\cite{Ngo:PODS:2012} is immediate:
\begin{cor}[\refthm{thm:lftj}]
\label{cor:agmlftj}
The run time of Leapfrog Triejoin is bounded by the
fractional edge cover bound, up to a log factor.
\end{cor}

Example: for the $R, S, T$ example we could define the family of
instances by $|R| \leq n$, $|S| \leq n$,
$|T| \leq n$; the fractional edge cover bound
provides $q(n)=n^{3/2}$, and therefore the running
time of leapfrog triejoin is $O(n^{3/2} \log n)$.

\section{Improving on the NPRR bound}

\label{s:finer}

In this section we show that leapfrog triejoin achieves optimal
worst-case running time (up to a log factor) over finer-grained
families of problem instances than those defined by the AGM (fractional
edge cover) bound, and that NPRR is \emph{not} worst-case optimal
for such families.

\subsection{Instances defined by projection bounds}

\label{s:fasterthannprr}

\refcor{cor:agmlftj} established that the leapfrog triejoin complexity
bound of $O(q(n) \log M(n))$ from \refthm{thm:lftj} applies to families
of problem instances defined by constraints on the sizes of input relations.
We demonstrate that \refthm{thm:lftj}
applies to finer-grained families defined by
constraints on the size of \emph{projections} of input relations.
This establishes that leapfrog triejoin is worst-case optimal
for such families.  The NPRR algorithm is not;
we exhibit a family of problem instances for which leapfrog triejoin
is optimal, and NPRR is asymptotically slower.  This may partially
explain the faster performance of LFTJ observed in practice 
\cite{Nguyen:private:2012}.

By way of example, we return to $Q(a,b,c)=R(a,b)$, $S(b,c)$, $T(a,c)$.
Consider a family of problem instances $(\mathbf{K}_n)_{n \in \omega}$
defined by the following constraints on projection sizes of $R,S,T$
(recall that $R(a,\_) = \pi_1(R)$):
\begin{align*}
\begin{array}{rclcrcl}
|R(a,\_)| &\leq & n^{3/8} &~~~~ & |R(\_,b)| &\leq & n^{5/8} \\
|S(b,\_)| &\leq & n^{5/8} & & |S(\_,c)| &\leq & n^{3/8} \\
|T(a,\_)| &\leq & n       & & |T(\_,c)| &\leq& 1
\end{array}
\end{align*}

From the definition of the renumbering transform (\refsec{s:renumbering}),
the following is evident:
\begin{prop}
Applying a renumbering transform to a relation $R$ does not
increase the cardinality of any projections of $R$.
\end{prop}

Therefore families defined by constraints on projection cardinalities
are closed under renumbering, and the following is immediate from
\refthm{thm:lftj}:
\begin{thm}
\label{thm:lftjproj}
Leapfrog triejoin is worst-case optimal for families of problem
instances defined by cardinality constraints on projections of
input relations.
\end{thm}

Continuing our example, from the above constraints on projections
of $R,S,T$
it is easily inferred that $|Q(a,\_,\_)|$ $\leq n^{3/8}$,
$|Q(\_,b,\_)| \leq n^{5/8}$, and $|Q(\_,\_,c)| \leq 1$.
Hence $|Q| \leq n$.
By \refthm{thm:lftjproj}, leapfrog triejoin runs in time
$O(n \log n)$ over this family.
In the next section we establish that NPRR has running
time $\Theta(n^{1.375})$ for this family.  This counterexample
establishes:

\begin{prop}
\label{prop:NPRRnot}
NPRR is \emph{not} worst-case optimal for families of problem
instances defined by cardinality constraints on projections of
input relations.
\end{prop}

\subsubsection{Counterexample for \refprop{prop:NPRRnot}}

Consider the behaviour of the NPRR algorithm for an instance with:
\begin{align*}
\begin{array}{lcllllll}
R &=& [n^{3/8}] & \times & [n^{5/8}] \\
S &=&           &        & [n^{5/8}] & \times & [n^{3/8}] \\
T &=& [n]       &        &           & \times & [1] \\
Q &=& [n^{3/8}] & \times & [n^{5/8}] & \times & [1]
\end{array}
\end{align*}

We follow the exposition of Example 2 of \cite{Ngo:PODS:2012}.
Let $\tau \geq 0$ be a parameter, which is used to
define a threshold for \emph{heavy join keys}.
A join key $b \in R(\_,b)$ is heavy if it appears
in more than $\tau$ tuples of $R$; let $D$ be the
set of heavy join keys.
The algorithm handles tuples $(a,b) \in R$ for heavy join keys
$b \in D$ separately from those with $b \not\in D$.
Let $G \subseteq R$ be those tuples containing no
$b \in D$.
The algorithm (1) 
constructs $D \times T$ and filters using hash tables
on $S$ and $R$; and (2) constructs
$G \bowtie S$ and filters using a hash table on $T$.
The union of these two results yields $Q$.

We now consider the running time.  There are two cases, which depend
on the choice of $\tau$ (which is taken to be $n^{1/2}$ in Example 2
of \cite{Ngo:PODS:2012}, but we consider arbitrary choice of $\tau$
here.)

Case 1: $\tau \geq n^{3/8}$.  Then $D$ will be empty, and $G=R$;
the result will be constructed using only step (2): constructing
$G \bowtie S = R \bowtie S$ and filtering using $T$.
Since $|R \bowtie S| = n^{1+3/8}$, the running time
will be $\Theta(n^{1.375})$.

Case 2: $\tau < n^{3/8}$.  Then $D = [n^{5/8}]$, and the result will be
constructed using only step (1): since $|D \times T| = n^{1+5/8}$,
the running time will be $\Theta(n^{1.625})$.

With the best choice of $\tau$ the running time is
$\Theta(n^{1.375})$.  Since leapfrog triejoin has running
time $O(n \log n)$ for the family of problem instances
containing this example, we have demonstrated that
leapfrog triejoin can be asymptotically faster than the NPRR
algorithm.

\section{Discussion and future Work}


\subsection{Removing the log factor}

\label{s:hashvariant}

Ken Ross suggested the following variant of leapfrog triejoin
which eliminates the $\log M(n)$ factor of the complexity bound
\cite{Ross:private:2012}.
For a relation $R(a,b)$, maintain a hash table for
$R(a,\_)$ i.e. for the projection $\pi_1(R)$.  For each
$a \in R(a,\_)$ maintain a hash table for $R_a(b)$.
Each entry in the hash table for $R(a,\_)$ contains a
pointer to the hash table for $R_a(b)$.  (Similarly for $k$-ary relations
with $k > 2$).
Replace each leapfrog join with a scan of the smallest
relation, with lookups into hash tables for the other relations.
This eliminates the log factor,
giving a running time of $O(q(n))$.

It will be interesting to investigate when
the asymptotic improvement offered by hash tables translates into
practical advantage, and whether query optimizers can be trained
to efficiently select trie versus hash table representations.
There are countervailing factors to be weighed against the
asymptotic improvement:
\begin{itemize}
\item The leapfrog join of unary relations $A_1,\ldots,A_k$ can require
substantially fewer than $\min \{ |A_1|,\ldots,|A_k| \}$ iterator operations
in practice,
due to differences in data distribution amongst the relations.  Using leapfrog join ensures that
you do not pay for the relation sizes per se, but rather for the
interleavings where one relation interposes itself into another.
One example of this is given in
\refsec{s:densityexample}, where a join of three relations of
size $n$ is performed with $O(1)$ iterator operations.
This advantage is not obviously achievable by the hash table
variant.
\item The $\log M(n)$ factor in the leapfrog triejoin
complexity bound is a tax not always applied; the log factor
reflects the potential cost of sparse access patterns into relations,
when leaping between distant keys.
When access patterns are dense, the log factor vanishes.
For example, taking $R=S=T=[n^{1/2}] \times [n^{1/2}]$, and representing
$R,S,T$ as tries, the running time is $O(n^{3/2})$.
\item Hash tables imply random memory access patterns,
which are notoriously costly in steep memory hierarchies; leapfrog
triejoin frequently exhibits sequential access patterns, which are better exploited
by current architectures.
\end{itemize}

\input{extensions}

\subsection*{Conclusions}

Leapfrog triejoin is a variable-oriented join algorithm
which achieves worst-case optimality (up to a log factor)
over large and useful families of problem instances.
It provides the core evaluation algorithms of the LogicBlox
Datalog system.
It improves on the NPRR algorithm in its simplicity,
and its optimality for finer-grained families of problem instances.
The algorithm is easily understood and straightforward
to implement.

\subsection*{Acknowledgements}

Our thanks to Dan Olteanu, Todd J. Green,
Kenneth Ross, Daniel Zinn and Molham Aref for
feedback on drafts of this paper.
Our gratitude to Dung Nguyen, whose
benchmarks comparing leapfrog triejoin to
the NPRR algorithm motivated this work.

\bibliography{bibliography}
\bibliographystyle{plain}

%% file: bound.tex
\subsection{The fractional cover bound}

\label{s:fractionalcover}

We begin with a review of the fractional edge cover bound
for worst-case result size
of full conjunctive queries.  The fractional edge cover bound
is not directly required by the complexity analysis for
leapfrog triejoin (\refthm{thm:lftj}), which is formulated in terms of
the maximum query result size $Q^\ast$.  However, for families
of problem instances defined by cardinality constraints on
the relation sizes, $Q^\ast$ can be computed using the
fractional cover bound.

The fractional edge cover was proven to be an upper bound
on query result size by
Grohe and Marx \cite{Grohe:SODA:2006} in the context of
constraint solving.  The bound was shown to be tight,
and adapted to relational joins, by Atserias, Grohe and
Marx \cite{Atserias:FOCS:2008}.

We continue the running example of a query $Q(a,b,c)$ defined
by the join $R(a,b)$, $S(b,c)$, $T(a,c)$.
Suppose we know the sizes $|R|$, $|S|$, and $|T|$,
and we wish to know the largest possible query result size $|Q|$.
The AGM bound for $|Q|$ is obtained by constructing a hypergraph 
$H=(V,\mathcal{E})$ whose
vertices are the variables $V=\{a,b,c\}$, and each atom
such as $R(a,b)$ is interpreted as a (hyper)edge $\{a,b\}$
on the variables appearing in its argument list:
\begin{align*}
\xymatrix @=1cm {
& \xynode{a} \ar@{-}[rd]^{T} \ar@{-}[dl]_{R} \\
\xynode{b} \ar@{-}[rr]_{S} & & \xynode{c}
}
\end{align*}
Recall that an edge cover is a subset $C \subseteq \mathcal{E}$
of edges such that each vertex appears in at least one edge $e \in C$.
Edge cover can be formulated as an integer
programming problem by assigning to each edge $e_i \in \mathcal{E}$
a weight $\lambda_i$, with $\lambda_i=1$ when $e_i \in C$
and $\lambda_i=0$ when $e_i \not\in C$.
The cover requirement is enacted by inequalities, one for
each vertex.  For the query (\ref{e:RST}) we would
use edge weights $\lambda_R$, $\lambda_S$, and $\lambda_T$,
and inequalities:
\begin{align}
\label{e:inequalities}
\begin{array}{cccccccc}
a: & \lambda_R & + & & & \lambda_T & \geq& 1  \\
b: & \lambda_R & + & \lambda_S & & &\geq& 1 \\
c: & & & \lambda_S & + & \lambda_T &\geq& 1 
\end{array}
\end{align}
A \emph{fractional edge cover} is obtained by relaxing to
a linear programming problem, permitting edge weights to
range between $0$ and $1$.  For example, choosing
$\lambda_R = \lambda_S = \lambda_T = \tfrac{1}{2}$
yields a valid fractional cover.  Grohe and Marx
\cite{Grohe:SODA:2006} established that:
\begin{align}
|Q| &\leq |R|^{\lambda_R} \cdot |S|^{\lambda_S} \cdot |T|^{\lambda_T}
\end{align}
Or equivalently:
\begin{align}
\log |Q| &\leq \lambda_R \log |R| + \lambda_S \log |S| + \lambda_T \log |T|
\label{e:RSTfraccover}
\end{align}
Minimizing the right-hand side of (\ref{e:RSTfraccover}) yields
the AGM bound on the size of $|Q|$.  For example, with
$|R|=|S|=|T|=n$, the bound is minimized when
$\lambda_R=\lambda_S=\lambda_T=\tfrac{1}{2}$,
yielding $|Q| \leq n^{3/2}$.

\subsection{Dual formulation}
\label{s:dual}
The dual formulation, used by \cite{Atserias:FOCS:2008} to
prove tightness of the bound, is more intuitive and offers
a construction of worst-case instances that is instructive.
We introduce the dual through an example.

Consider a scenario where the sizes of R, S, T are fixed, and
$Q(a,b,c)$ has a simple cross-product
structure $Q=[2^\alpha] \times [2^\beta] \times [2^\kappa]$.
(The quantities $\alpha,\beta,\kappa$ can be interpreted
as the average number of bits to represent variables $a,b,c$;
for simplicity we assume $2^\alpha,2^\beta,2^\kappa$ to
be integers.)
From the query definition (\ref{e:RST}), it is apparent that
$(a,b,c) \in Q$ implies $(a,b) \in R$; therefore
$[2^\alpha] \times [2^\beta] \subseteq R$.
This implies $\alpha + \beta \leq \log |R|$.  Similarly
for $S$ and $T$.  The
problem of maximizing $|Q|$ can be formulated as a
linear program:
\begin{align*}
\begin{array}{ll}
\mbox{Maximize} & \log |Q| = \alpha + \beta + \kappa \\
\\
\mbox{Subject to} &
\left\{
\begin{array}{ccccccc}
\alpha & + & \beta & & &\leq& \log |R| \\
& & \beta & + & \kappa &\leq& \log |S| \\
\alpha & + & & & \kappa & \leq& \log |T|
\end{array}
\right.
\end{array}
\end{align*}
For example, setting $\log |R| = \log |S| = \log |T| = \log n$
yields $\log |Q| = \tfrac{3}{2} \log n$ at optimality, achieved
by $\alpha=\beta=\kappa=\tfrac{1}{2} \log n$
and $Q=[n^{1/2}]\times[n^{1/2}]\times[n^{1/2}]$.

The above linear program is the dual of the
fractional edge cover linear program:
using $\boldsymbol{\lambda}=[\lambda_A,\lambda_B,\lambda_C]$,
$\boldsymbol{\eta}=[\log |R|, \log |S|, \log |T|]$,
$\boldsymbol{\alpha}=[\alpha,\beta,\kappa]$, and
$\boldsymbol{1}=[1,1,1]$, the fractional edge
cover program minimizes $\boldsymbol{\eta}^\top \boldsymbol{\lambda}$
subject to $\boldsymbol{A} \boldsymbol{\lambda} \geq \boldsymbol{1}$
(each row of $\boldsymbol{A}$ yielding an inequality of \refeqn{e:inequalities})
and $\boldsymbol{\lambda} \geq 0$;
the dual form
maximizes $\boldsymbol{1}^\top \boldsymbol{\alpha}$
subject to $\boldsymbol{A}^\top \boldsymbol{\alpha} \leq \boldsymbol{\eta}$
and $\boldsymbol{\alpha} \geq 0$.
It follows from the duality property of linear programs that an
optimal solution to the dual form yields the same
upper bound on $|Q|$ as the optimal fractional edge cover.

Moreover, the dual form is constructive:
let $n_a=\lfloor 2^\alpha \rfloor$,
$n_b=\lfloor 2^\beta \rfloor$,
and $n_c=\lfloor 2^\kappa \rfloor$, and choose
\begin{align*}
R \supseteq [n_a]\times[n_b] \\
S \supseteq [n_b]\times[n_c] \\
T \supseteq [n_a]\times[n_c]
\end{align*}
padding with rubbish as necessary to attain the
desired sizes $|R|$, $|S|$, and $|T|$.
This yields a $Q$ of maximal size.

This construction prompts the following observation:
the worst cases of the AGM bound are achievable by
query results which are cross-products.
Since real-world queries rarely have such
a structure---practical database systems
avoid materializing such queries---this
suggests that an algorithm achieving
the AGM bound is not necessarily optimal for classes
of database instances encountered in practice.
This motivates our development of finer-grained
classes in \refsec{s:finer}.

%% file: extensions.tex
\subsection{Extension to $\exists_1$ queries}

\label{s:extensions}

The implementation of leapfrog triejoin in our commercial database system
LogicBlox extends the basic algorithm described here in several useful
ways.  We sketch these extensions here, as they provide basic functionality
essential for implementors.

The LogicBlox runtime evaluates rules defined using the following
fragment of first-order logic:
\begin{align*}
\mathsf{conj} &::= ~[~ \exists \overline{x} ~.~ ~]~ \mathsf{dform} ~(~\wedge \mathsf{dform} ~)^\ast
~ \\
\mathsf{dform} &::= \mathsf{atom} ~|~ \mathsf{disj} ~|~ \mathsf{negation} \\
\mathsf{atom} &::= R(\overline{y}) ~|~ F[\overline{y}]=\overline{z} \\
\mathsf{disj} &::= \mathsf{conj} ~(~ \vee \mathsf{conj} )^+ \\
\mathsf{negation} &::= \neg \mathsf{conj} \\
\mathsf{rule} &::= \forall \overline{x} ~.~ \mathsf{conj} \rightarrow \mathsf{head} \\
\mathsf{head} &::= \mathsf{atom} ~(~ \wedge \mathsf{atom} )^\ast
\end{align*}
Each conjunction bears an optional existential quantifier block.
Atoms can be either relations or functions, which may represent
either concrete data structures (representing edb functions/relations,
or materialized views),
or primitives such as addition and multiplication.
The use of negation comes with some further restrictions not
captured by the above grammar.
We extend leapfrog triejoin to tackle such rules as follows.
\begin{enumerate}
\item \emph{Disjunctions.}  
A simple variant of the leapfrog algorithm
computes a disjunction of unary relations $A_1(x) \vee \cdots \vee A_k(x)$,
using the standard algorithm for merging sorted sequences presented
by iterators.
We handle a disjunction of $k$-ary formulas
$\varphi_1(\overline{x}) \vee \cdots \vee \varphi_k(\overline{x})$
in the following manner.  Each subformula $\varphi_i(\overline{x})$
is required to have the same free variables.
The extension from disjunction of unary subformulas to $k$-ary subformulas
mostly follows the triejoin algorithm of \refsec{s:triejoinimplementation},
with the exception that the triejoin-open() method selects only
those iterators positioned at the current key for opening at the
next level.  The implementation of disjunction presents 
$\varphi_1(\overline{x}) \vee \cdots \vee \varphi_k(\overline{x})$
as a nonmaterialized view using the trie iterator interface;
since leapfrog triejoin likewise presents conjunctions
as nonmaterialized views, we can permit arbitrary nesting of
conjunctions and disjunctions without materializing intermediate
results or DNF-conversion, in most cases.
\item \emph{Functions.}  We distinguish between relations $R(x,$ $y)$
and functions $F[x]=y$.  For the function $F[x]=y$, where $F$ is
represented by a concrete data structure, the variable
$x$ is said to occur in \emph{key position}, and the variable $y$
in \emph{value position}.  A free variable of a conjunction
is a key if it is a key of any subformula;
every free variable of a disjunction is deemed to be a key.
Only variables occurring in key position 
are considered for the variable ordering.  To handle a query
such as $F[x]=y,G[y]=z$ with a variable ordering $[x,y]$,
we treat it as $F[x]=\alpha,\mathit{I}_\alpha(y),G[y]=z$,
where $\mathit{I}_\alpha(y)$ presents a nonmaterialized view of
the relation $\{ \alpha \}$.
\item \emph{Primitives.}
Primitives are scalar operations such as addition and multiplication.
In a subformula such as $z=y+1$, we deem all variable occurrences
to be value-position.  
In a query such as
$A(x,y),z=y+1$, we handle primitives such as $z=y+1$
by attaching \emph{actions} to the
triejoin which are triggered whenever a specified variable
is bound.  With the key ordering $[x,y]$, the primitive
$z=y+1$ would be triggered whenever $y$ is bound.
We order actions attached to the same variable
so as to respect order-of-operation dependencies.
For a query such as $A(x,y),z=y+1,B(z)$ with variable
ordering $[x,y,z]$, we use the same technique as above, treating it as
$A(x,y),\alpha=y+1,I_\alpha(z),B(z)$.
Actions can either succeed or fail; if they fail,
the leapfrog triejoin algorithm searches for the
next binding of the trigger variable.
\item \emph{Negation.}  We distinguish two cases of negation:
complementation, and scalar negation.  Complementation occurs
when we have a formula of the form
$\varphi_1(\overline{x}$, $\overline{y}),\neg \varphi_2(\overline{x})$,
where each variable in $\overline{x}$ occurs in key-position of
$\varphi_1$.  In this case we handle $\neg \varphi_2(\overline{x})$
by an action attached to the last variable of $\overline{x}$,
which performs a lookup into $\varphi_2(\overline{x})$,
succeeding just when $\varphi_2(\overline{x})$ fails.
Scalar negation occurs when we have a
subformula $\neg \varphi_2(\overline{x})$, where each
variable in $\overline{x}$ occurs in value-position in
$\varphi_2$; this implies $\varphi_2$ contains only
primitive operations.  In such cases we permit existential
quantifiers to occur in $\varphi_2$, to handle subexpressions
such as $\neg \exists t ~.~ x+y=t,t > 0$, which depart from
$\exists_1$ in a trivial way.  We handle these by an action attached
to the last variable of $\overline{x}$ which computes
$\varphi_2(\overline{x})$ and succeeds just when
$\varphi_2$ fails.
\item \emph{Projections.}  Projections are currently handled
by using data structures which support reference cou\-nts.  
we anticipate introducing an optimization to handle a projection
$\exists z ~.~ \varphi(\overline{x},z)$ by a special
nonmaterialized view which produces only the first
$z$ for given $\overline{x}$.  We anticipate this will
be more efficient when $z$ occurs after $\overline{x}$
in the variable ordering.
\item \emph{Ranges.}  We handle inequalities such as $x \geq c$
by including a nonmaterialized view of a predicate representing
an interval $[c,+\infty)$; similarly for $\leq, <, >$.  It is a
simple exercise to implement a trie-iterator for an interval
such as $[c,+\infty)$.
\end{enumerate}

The complexity analysis presented for leapfrog triejoin
(\refthm{thm:lftj}) does not immediately encompass the
above extensions, but can be applied in some cases
by considering nested subformulas to be materialized, even though
in actual evaluation they are not.  For example,
in the formula $A(x,y),(B(y,z) \vee C(y,z))$,
we can consider a hypothetical materialization
of $T(y,z) \equiv B(y,z) \vee C(y,z)$, and
analyze $A(x,y),T(y,z)$ using whatever properties
for $T$ we can establish.  For instance, if
we know $|A| \leq n_1$, $|B| \leq n_2$ and $|C| \leq n_3$,
it follows that $|T| \leq (n_2+n_3)$,
and we can invoke \refcor{cor:agmlftj} using
the fractional cover bound.  This technique yields a
valid bound for the nonmaterialized presentation of $T$
if trie iterator operations on the presentation of
$B(y,z) \vee C(y,z)$ take $O(\log n)$ time,
which is the case for disjunctions.  However, this is
not the case for projections.  Laurent Oget made the
promising suggestion of lazily materializing 
subformulas as they are evaluated, which would
limit the cost to that of materializing
all subexpressions.